\documentclass[a4paper,10pt]{article}

\usepackage{figlatex}
\graphicspath{{fig/}}
\usepackage[latin1]{inputenc}
\usepackage[american]{babel}
\usepackage{enumerate,paralist}
\usepackage{amsfonts,amsmath,amssymb,amstext,latexsym, amsthm}
\usepackage{url,xspace}
\usepackage[colorinlistoftodos]{todonotes}
\usepackage{a4wide}
\usepackage[ruled,vlined]{algorithm2e}
\SetAlgoSkip{}

\newtheorem{theorem}{Theorem}
\newtheorem{lemma}[theorem]{Lemma}

\newcommand{\ie}{{\it i.e.}\xspace}

\newcommand{\card}[1]{\ensuremath{|{#1}|}\xspace}

\newcommand{\simplerandom}{\textsc{Random\-Outer}\xspace}
\newcommand{\simplesorted}{\textsc{Sorted\-Outer}\xspace}
\newcommand{\stupid}{\textsc{Dynamic\-Outer}\xspace}
\newcommand{\stupidthreshold}{\textsc{Dynamic\-Outer\-2Phases}\xspace}

\newcommand{\simplerandommat}{\textsc{Random\-Matrix}\xspace}
\newcommand{\simplesortedmat}{\textsc{Sorted\-Matrix}\xspace}
\newcommand{\stupidmat}{\textsc{Dynamic\-Matrix}\xspace}
\newcommand{\stupidthresholdmat}{\textsc{Dynamic\-Matrix\-2Phases}\xspace}

\newcommand{\mapreduce}{MapReduce\xspace}

\newcommand{\ema}[1]{\ensuremath{#1}\xspace}
\newcommand{\dx}{\ema{\delta x}}

\begin{document}

\title{Analysis of Dynamic Scheduling Strategies for Matrix Multiplication on Heterogeneous Platforms}

\author{Olivier Beaumont\footnote{Inria \& University of Bordeaux (\protect\url{olivier.beaumont@inria.fr})} 
  \and 
  Loris Marchal \footnote{LIP (CNRS, INRIA, ENS-Lyon, Univ. of Lyon), (\protect\url{loris.marchal@ens-lyon.fr})}}

\maketitle

\begin{abstract} 
  The tremendous increase in the size and heterogeneity of
  supercomputers makes it very difficult to predict the performance of
  a scheduling algorithm. Therefore, dynamic solutions, where
  scheduling decisions are made at runtime have overpassed static
  allocation strategies. The simplicity and efficiency of dynamic
  schedulers such as Hadoop are a key of the success of the MapReduce
  framework. Dynamic schedulers such as StarPU, PaRSEC or
  StarSs are also developed for more constrained computations, e.g. 
  task graphs coming from linear algebra. To make their decisions,
  these runtime systems make use of some static information, such as
  the distance of tasks to the critical path or the affinity between
  tasks and computing resources (CPU, GPU,\ldots) and of dynamic
  information, such as where input data are actually located. In this
  paper, we concentrate on two elementary linear algebra kernels, namely the
  outer product and the matrix multiplication. For each problem, we
  propose several dynamic strategies that can be used at runtime and
  we provide an analytic study of their theoretical performance. We
  prove that the theoretical analysis provides very good estimate of
  the amount of communications induced by a dynamic strategy and can
be used in order to efficiently determine thresholds used in dynamic scheduler, thus
  enabling to choose among them for a given problem and architecture.
\end{abstract}

\section{Introduction}
\label{intro}

Recently, there has been a very important change in both parallel
platforms and parallel applications. On the one hand, computing
platforms, either clouds or supercomputers involve more and more
computing resources. This scale change poses many problems, mostly
related to unpredictability and failures. 
Due to the size of the platforms, their complex
network topologies, the use of heterogeneous resources, NUMA effects,
the number of concurrent simultaneous computations and communications,
it is impossible to predict exactly the time that a specific task will
take. Unpredictability makes it impossible to statically allocate the
tasks of a DAG onto the processing resources and dynamic scheduling and allocation strategies
are needed. As a consequence, in recent years, there has been a large
amount of practical work to develop efficient runtime schedulers. The main
characteristics of these schedulers is that they make their decisions
at runtime, based on the expected duration of the tasks on the
different kind of processing units (CPUs, GPUs,...) and on the
expected availability time of the task input data, given their actual
locations. Thanks to these information, the scheduler allocates the
task to the resource that will finish its processing as soon as possible.
Moreover, all these runtime systems also make use of some
static information that can be computed from the task graph itself, in
order to decide the priority between several ready tasks. This
information mostly deals with the estimated critical path as proposed
in HEFT~\cite{heft} for instance.

On the other hand, there has been a dramatic simplification of the
application model in many cases, as asserted by the success of the
\mapreduce framework~\cite{dean2008mapreduce} which has been
popularized by Google. It allows users without particular knowledge in
parallel algorithms to harness the power of large parallel machines. In
\mapreduce, a large computation is broken into small tasks that run in
parallel on multiple machines, and scales easily to very large
clusters of inexpensive commodity computers. \mapreduce is a very
successful example of dynamic schedulers, as one of its crucial
feature is its inherent capability of handling hardware failures and
processing capabilities heterogeneity, thus hiding this complexity to
the programmer, by relying on on-demand allocations and the on-line detection of
nodes that perform poorly (in order to re-assign tasks that slow down
the process). As we explained in a previous work~\cite{nofreelunch},
\mapreduce, although tailored for linear complexity operations (such as
text parsing), is now widely used for non linear complexity tasks. In
this case, it induces a large replication of the data. For example,
when \mapreduce is used to compute the outer product of two vectors
$a$ and $b$, the most common technique is to emit all possible pairs
of $(a_i,b_j)$, so that many processors can be used to compute the
elementary products. This induces a large replication factor, since
\mapreduce is not aware of the 2-dimensional nature of the data.



Our goal in this paper is to show how simple data-aware dynamic
schedulers can be proven efficient in a specific context. We
concentrate here on two elementary kernels, namely the outer product
and the matrix multiplication. These kernels do not induce
dependencies among their tasks, but because of their massive input data reuse results,
a straightforward \mapreduce implementations of these
kernels would involve a large replication overhead.  Indeed, in both
cases~\cite{nofreelunch}, input vectors or input matrices need to be
replicated when the kernel is processed by a large-scale parallel
platform, and basic dynamic strategies that allocate tasks at random
to processors fail to achieve reasonable communication volumes with
respect to known lower bounds.

In the present paper, we first present and study a very simple yet
efficient dynamic scheduler for the outer
product, that generates a communication volume close to the lower
bound. Our main contribution is to analyze the communication volume
generated by the dynamic scheduler as a continuous process that can be
modeled by an Ordinary Differential Equation (ODE). We prove that the
analytic communication volume of the solution of the ODE is close to
the actual communication volume as measured using
simulations. Moreover, we prove that this analysis of the solution of
the ODE can be used in order to optimize a dynamic randomized allocation strategy, for
instance, by switching between two strategies when the number of
remaining tasks is smaller than a given threshold, that is determined by the theoretical analysis. This simple example
attests the practical interest of the theoretical analysis of dynamic
schedulers, since it shows that the analytic solution can be used in
order to incorporate static knowledge into the scheduler.  After
presenting our method on the outer product (Section~\ref{outer}),
we move to a more common kernel, the matrix multiplication and show how the
previous analysis can be extended in Section~\ref{matrix_mult}.

\section{Related work}

We briefly review previous works related to our study, which deals both
with actual runtime schedulers and with their theoretical studies.

\subsection{Runtime dynamic schedulers}

As mentioned in the introduction, several runtime systems have been
recently proposed to schedule applications on parallel systems. Among
other successful projects, we may cite StarPU~\cite{starpu}, from INRIA
Bordeaux (France), DAGuE and PaRSEC~\cite{dague,parsec} from ICL,
Univ. of Tennessee Knoxville (USA) StarSs~\cite{starss} from Barcelona
Supercomputing Center (Spain) or KAAPI~\cite{kaapi} from INRIA Grenoble
(France). Most of these tools enable, to a certain extent, to schedule
an application described as a task graph (usually available in the
beginning of the computation, but sometimes generated and discovered
during the execution itself), onto a parallel platforms. Most of these tools
allow to harness complex platforms, such as multicores and hybrid
platforms, including GPUs or other accelerators. These runtime systems
usually keep track of the occupation of each computing devices and
allocate new tasks on the processing unit that is expected to minimize
its completion time. Our goal in this paper in to provide an analysis
of such dynamic schedulers for simple operations, that do not involve tasks dependencies but massive data reuse.

\subsection{Theoretical studies of dynamic systems}

Many studies have proposed to use queuing
theory~\cite{queueing-theory-book} to study the behavior of simple
parallel systems and their dynamic evolution. Among many others, Berten
et al.~\cite{BertenG07} propose to use such stochastic models in order to model
computing Grids, and Mitzenmacher~\cite{Mitzenmacher00oldinfo} studies
how not-to-date information can lead to bad scheduling decisions in a
simple parallel system.

Recently, mean field techniques~\cite{gast2012mean,benaim2008class}
have been proposed for analyzing such dynamic processes. They give a
formal framework to derive a system of ordinary differential equations
that is the limit of a Markovian system when the number of objects
goes to infinity. Such techniques have been used for the first time
in~\cite{Mitzenmacher1998loadstealing} where the author derives
differential equations for a system of homogeneous processors who
steal a single job when idle.


\section{Randomized dynamic strategies for the outer-product}
\label{outer}

We present here the analysis of a dynamic scheduler for a simple
problem from linear algebra, namely the outer-product of two vectors.

\subsection{Problem definition}

We consider the problem of computing the outer-product $a b^t$ of two
large vectors $a$ and $b$ of size $N$, \ie to compute all values $a_i
\times b_j, \forall 1 \leq i,j \leq N$. The computing domain can
therefore be seen as a matrix. For granularity reasons, we will
consider that $a$ and $b$ are in fact split into $N/l$ blocks of size
$l$ and that a basic operation consists in computing the outer product
of two (small) vectors of size $l$.

As stated above, we target heterogeneous platforms consisting of $p$
processors $P_1,\ldots,P_p$, where the speed of processor $P_i$, \ie
the number of outer products of size $l$ vectors that $P_k$ can do in
one time unit, is given by $s_k$. We will also denote by
$\mathit{rs}_k$ the relative speed of $\mathit{rs}_k =
\frac{s_k}{\sum_i s_i}$. Note that the randomized strategies that we
propose are agnostic to processor speeds, but they are demand driven,
so that a processor with a twice larger speed will request work twice
faster.

In the following, we assume that a master processor coordinates the
work distribution: it is aware of which $a$ and $b$ blocks are
replicated on the computing nodes and decides which new blocks are
sent, as well as which tasks are allocated to the
nodes. After completion of their allocated tasks, computing nodes
simply report to the master processor, requesting for new tasks.

We will assume throughout the analysis that it is possible to overlap
computations and communications. This can be achieved with dynamic
strategies by uploading a few blocks in advance at the beginning of
the computations and then to request work as soon as the number of
blocks to be processed becomes smaller than a given
threshold. Determining this threshold would require to introduce a
communication model and a topology, what is out of the scope of this
paper, and we will assume that the threshold is known. In practice,
the number of tasks required to ensure a good overlap has been observed to be small
in~\cite{kreaseck2003autonomous,parashar2006autonomic} even though a
rigorous algorithm to estimate it is still missing.

As we observed~\cite{nofreelunch}, performing a non linear complexity
task such as a Divisible Load or a MapReduce operation requires to
replicate initial data. Our objective is to minimize the overall
amount of communications, \ie the total amount of data (the number of
blocks of $a$ and $b$) sent by the master node initially holding the data,
or equivalently by the set of devices holding the data since we are
interested in the overall volume only, under the constraint that a
perfect load-balancing should be achieved among resources allocated to
the outer product computation. Indeed, due to data dependencies, if we
were to minimize communications without this load-balancing
constraint, the optimal (but very inefficient) solution would consist in
making use of a single computing resource so that each data block would be
sent exactly once.

\subsection{Design of randomized dynamic strategies}

As mentioned above, vectors $a$ and $b$ are split into $N/l$ data
blocks. In the following, we denote by $a_i$ the $i$th block of $a$
(rather than the $i$th element of $a$) since we always consider
elements by blocks. As soon as a processor has received two data
blocks $a_i$ and $b_j$, it can compute the block $M_{i,j}=(a
b^t)_{i,j}=a_i b_j^t$. This elementary task is denoted
by $T_{i,j}$. All data blocks are initially available at the master
node only.

One of the simplest strategy to allocate computational tasks to
processors is to distribute tasks at random: whenever a processor is
ready, a task $T_{i,j}$ is chosen uniformly at random among all
available tasks and is allocated to the processor. The data
corresponding to this task that is not yet on the processor, that is
one or two of the $a_i$ and $b_j$ blocks are sent by the master.  We
denote this strategy by \simplerandom.
Another simple option is to allocate tasks in
lexicographical order of indices $(i,j)$ rather than randomly. This strategy will be denoted as
\simplesorted.



Both previous algorithms are expected to induce a large amount of
communications because of data replication. Indeed, in these
algorithms, there is no reason why the data sent for the processing of
tasks on a given processor $P_k$ may be re-used for upcoming
tasks. This is why dynamic data-aware strategies have been
introduced. In the runtime systems cited above, such as StarPU, the
scheduler is aware of the locality of the data and uses this
information when allocating tasks to processors: it is much more
beneficial, when allocating a new task on $P_k$, to take advantage of
the $a$ and $b$ data already present on the processor, and to compute
for example all possible products $a_i b_{j'}^t$ before sending new
blocks of data. We propose such a strategy, denoted \stupid, in
Algorithm~\ref{algo.stupid}: when a processor $P_k$ receives a new
pair of blocks $(a_i, b_j)$, all possible products $a_i b_{j'}^t$ and
$a_{i'} b_j^t$ are also allocated to $P_k$, for all data blocks
$a_{i'}$ and $b_{j'}$ that have already been transmitted to $P_k$ in
previous steps.

\begin{algorithm}[htbp]
  \DontPrintSemicolon
  \While{there are unprocessed tasks}{    
    Wait for a processor $P_k$ to finish its tasks\;
    $I \gets \{i\textnormal{~such that~}P_k\textnormal{~owns~}a_i\}$\;
    $J \gets \{j\textnormal{~such that~}P_k\textnormal{~owns~}b_j\}$\;
    Choose $i\notin I$ and $j\notin J$ uniformly at random\;
    Send $a_i$ and $b_j$ to $P_k$\;
    Allocate all tasks of $\{T_{i,j} \} $ $\cup$ $\{T_{i,j'}, j'\in
    J\}$ $ 
    \cup $ $ \{ T_{i',j}, i'\in I\}$ that are not yet processed to $P_k$
    and mark them processed\;
  }  
  \caption{\stupid strategy.}
  \label{algo.stupid}
\end{algorithm}

Note that the \stupid scheduler is not computationally expensive: it
is sufficient to maintain a set of unknown $a$ and $b$ data (of size
$O(N/l)$) for each processor, and to randomly pick an element of this
set when allocating new blocks to a processor $P_k$. 

\begin{figure}[htbp]
  \centering
  \includegraphics[width=0.7\linewidth]{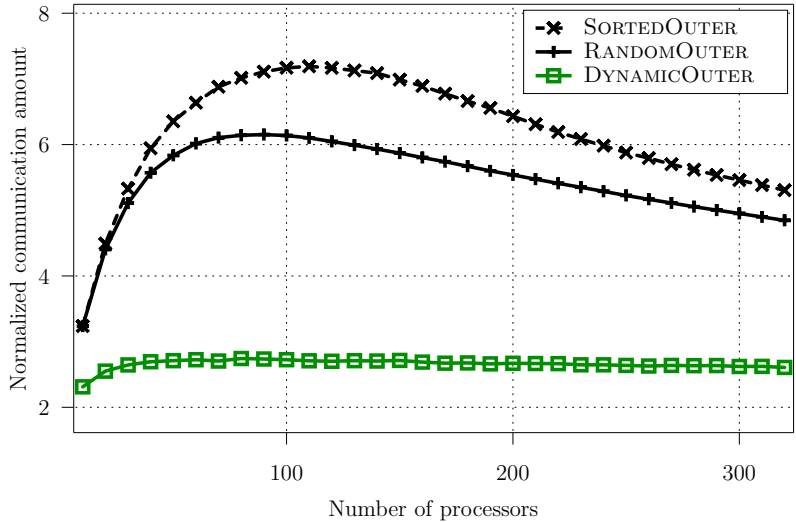}
  \caption{Comparison of random and data-aware dynamic strategies, for
    vectors of size $N/l=100$ blocks}
  \label{fig.rand_vs_data_aware}
\end{figure}

We have compared the performance of previous schedulers through
simulations on Figure~\ref{fig.rand_vs_data_aware}. Processor speeds
are chosen uniformly in the interval $[10, 100]$, which means a large
degree of heterogeneity. Each point in this figure and the following
ones is the average over 10 or more simulations. The standard
deviation is always very small, typically smaller than 0.1 for any
point, and never impacts the ranking of the strategies. It is thus not
depicted for clarity reasons. All communication amounts are normalized
with the following lower bound:
$$
LB=2 N \sum_k\sqrt{\mathit{rs}_k} = 2 N \sum_k\sqrt{\frac{s_k}{\sum_i s_i}},
$$ 
where $s_k$ is the speed of processor $P_k$ and $\mathit{rs}_k$ its
relative speed.

Indeed, in a very optimistic setting, each
processor is dedicated to computing a ``square'' area of $M= a b^t$,
whose area is proportional to its relative speed, so that all processors finish their work at the same instant. In this situation,
the amount of communications for $P_k$ is proportional to the half perimeter of this
square of area $N^2\mathit{rs}_k$. Note that this lower bound is not
expected to be achievable (consider for instance the case of 2 heterogeneous processors). The best known static algorithm (based on a complete knowledge of all relative speeds) has an
approximation ratio of $7/4$~\cite{algorithmica02}. This algorithm
computes an allocation scheme based on the computing speeds of the
processors. As outlined in
the introduction, such an allocation mechanism is not practical in our
context, since our aim is to rely on more dynamic runtime strategies, but can be used as a  comparison basis.

As expected, we notice on Figure~\ref{fig.rand_vs_data_aware} that
data-aware strategies induce significantly less communication than
purely random strategies.

Our \stupid allocation scheme suffers some limitation: when the number
of remaining blocks to compute is small, the proposed strategy is
inefficient as it may send a large number of $a$ and $b$ blocks to a
processor $P_k$ before it is able to process one of the last few
available tasks. Thus, we propose an improved version
\stupidthreshold in Algorithm~\ref{algo.stupidthreshold}: when the
number of remaining tasks becomes smaller than a given threshold, we switch to
the basic randomized strategy: any available task $T_{i,j}$ is allocated to a
requesting processor, without taking data locality into account. The
corresponding data $a_i$ and $b_j$  are then sent to $P_k$ if needed.

\begin{algorithm}[htbp]
  \DontPrintSemicolon
  \While{the number of processors is larger than the threshold}{    
    Wait for a processor $P_k$ to finish its tasks\;
    $I \gets \{i\textnormal{~such that~}P_k\textnormal{~owns~}a_i\}$\;
    $J \gets \{j\textnormal{~such that~}P_k\textnormal{~owns~}b_j\}$\;
    Choose $i\notin I$ and $j\notin J$ uniformly at random\;
    Send $a_i$ and $b_j$ to $P_k$\;
    Allocate all tasks of $\{T_{i,j} \} $ $\cup$ $\{T_{i,j'}, j'\in
    J\}$ $ 
    \cup $ $ \{ T_{i',j}, i'\in I\}$ that are not yet processed to $P_k$
    and mark them processed\;
  }  
  \While{there are unprocessed tasks}{    
    Wait for a processor $P_k$ to finish its tasks\;
    Choose randomly an unprocessed task $T_{i,j}$\;
    \lIf{$P_k$ does not hold $a_i$}{send $a_i$ to $P_k$}
    \lIf{$P_k$ does not hold $b_j$}{send $b_j$ to $P_k$}
    Allocate $T_{i,j}$ to $P_k$\;
  }
  \caption{\stupidthreshold strategy.}
  \label{algo.stupidthreshold}
\end{algorithm}

As illustrated on Figure~\ref{fig.threshold}, for a well chosen number
of tasks processed in the second phase, this new strategy allows to
reduce further the amount of communications. However, this requires to
accurately set the threshold, depending on the size of the matrix and
the relative speed of the processors. If too many tasks are processed
in the second phase, the performance is close to the one of
\simplerandom. On the contrary, if too few tasks are processed in the
second phase, the behavior becomes close to \stupid. The optimal
threshold corresponds here to a few percent of tasks being processed
in the second phase. In the following, we present an analysis of the
\stupidthreshold strategy that both allows to predict its performance
and to optimally set the threshold, so as to minimize the amount of
communications.

\begin{figure}[htbp]
  \centering
  \includegraphics[width=0.7\linewidth]{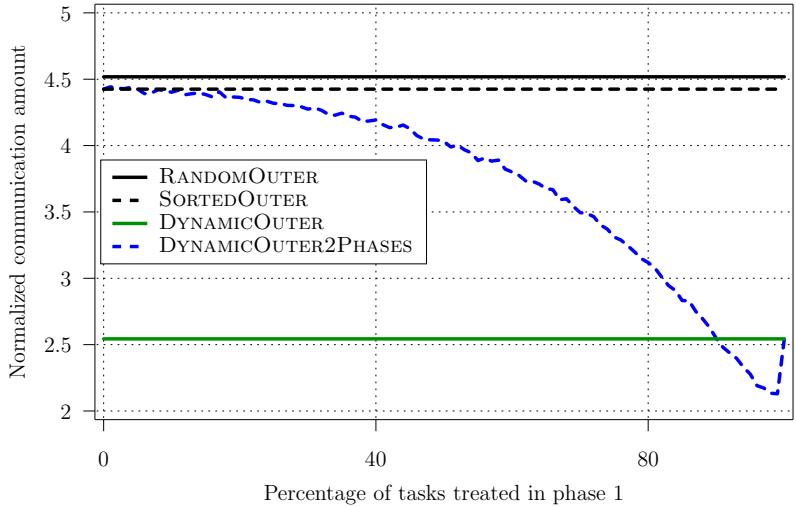}
  \caption{Communication amount of \stupidthreshold and comparison to
    the other schedulers for different thresholds (for a given
    distribution of computing speeds with 20 processors and
    $N/l=100$).}
  \label{fig.threshold}
\end{figure}

\subsection{Theoretical analysis of dynamic randomized strategies} 
\label{outer_analysis}

In this section, our aim is to provide an analytical model for
Algorithm \stupidthreshold. Analyzing such a strategy is crucial in order to
assess the efficiency of runtime dynamic strategies and in order to
tune the parameters of dynamic strategies or to choose among different
strategies depending on input parameters.

In what follows, we assume that $N$, the size of vectors $a$ and
$b$, is large and we consider a continuous dynamic process whose
behavior is expected to be close to the one of \stupidthreshold. In what
follows, we concentrate on processor $P_k$ whose speed is $s_k$.
At each step, \stupidthreshold chooses to send one data block of $a$ and one
data block of $b$, so that $P_k$ knows the same number of data blocks
of $a$ and $b$. As previously, we denote by $M=a b^t$ the result of
the outer product and by $T_{i,j}$ the tasks that corresponds to the
product of data blocks $a_i$ and $b_j$


We denote by $x=y/N$ the ratio of elements of $a$ and $b$ that are
known by $P_k$ at a given time step of the process and by
$t_k(x)$ the corresponding time step. We concentrate on a basic step
of \stupidthreshold during which the fraction of data blocks of both $a$ and
$b$ known by $P_k$ goes from $x$ to $x + \dx$. In fact, since \stupidthreshold
is a discrete process and the ratio known by $P_k$ goes from $x=y/N$
to $x+l/N = y/N + l/N$. Under the assumption that $N$ is large, we
assume that we can approximate the randomized discrete process by the
continuous process described by the corresponding Ordinary
Differential Equation on expected values. The proof of convergence is
out of the scope of this paper but we will show that this assumption
provides very good results through simulations in Section~\ref{simuouter}.

\begin{figure}[htbp]
  \centering
  \includegraphics[width=0.4\linewidth]{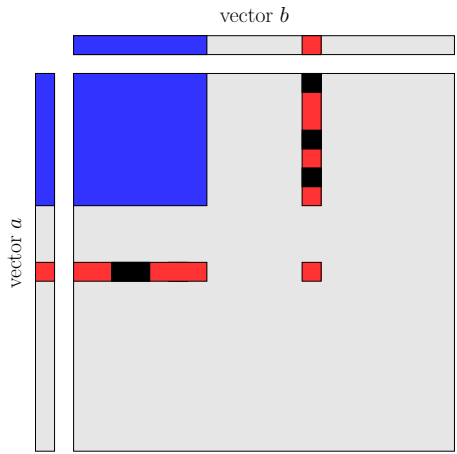}
  \caption{Illustration for the proof of Lemma~\ref{lemg}. The
    top-left blue  rectangle represents the data owned by
  the processor at time $t_k(x)$ (a permutation of the rows and
  columns has been applied to have it in the upper left corner). The new elements \dx are depicted in red, as well as the
  corresponding available tasks. Note that some tasks (in black)
  corresponding to the combination of \dx with the known elements have
already been processed by other processors.}
  \label{fig.evolv}
\end{figure}

Let us remark that during the execution of \stupidthreshold, tasks $T_{i,j}$
are greedily computed as soon as a processor knows the corresponding
data blocks of $a_i$ and $b_j$.  Therefore, at time $t_k(x)$, all
tasks $T_{i,j}$ such that $P_k$ knows data blocks $a_i$ and $b_j$ have
been processed and there are $x^2 N^2 / l^2$ such tasks. Note also
that those tasks may have been processed either by $P_k$ or by another
processor $P_j$ since processors compete to process tasks. Indeed,
since data blocks of $a$ and $b$ are possibly replicated on
several processors, then both $P_k$ and $P_\ell$ may know at some point
both $a_i$ and $b_j$. In practice, the processor which computes
$T_{i,j}$ is the one that learns both $a_i$ and $b_j$ first.

Figure~\ref{fig.evolv} depicts the computational domain during the
first phase of \stupidthreshold from the point of view of a given
processor $P_k$ (rows and columns have been reordered for the sake of
clarity). The top-left square (in blue) corresponds to value of $a$
and $b$ that are known by $P_k$, and all corresponding tasks have
already been processed (either by $P_k$ or by another processor). The
remaining ``L''-shaped area (in grey) corresponds to tasks $T_{i,j}$
such that $P_k$ does not hold either the corresponding value of $a$,
or the corresponding value of $b$, or both. When receiving a new value
of $a$ and $b$ (in red), $P_k$ is able to process all the tasks (in
red) from the two corresponding row and column. Some elements from
this row and this column may be already processed (in black).

In what follows, we denote by $g_k(x)$ the fraction of tasks $T_{i,j}$
in the previously described ``L''-shaped area that have not been
computed yet. We also assume that the distribution of unprocessed
tasks in this area is uniform, and we claim that this assumption is
valid for a reasonably large number of processors. Our simulations
below show that this leads to a very good accuracy.

Based on this remark, we are able to prove the following Lemma
\begin{lemma}
\label{lemg}
$g_k(x)=(1-x^2)^{\alpha_k}$, where $\alpha_k = \frac{\sum_{i \neq k} s_i}{s_k}$.
\end{lemma}

\begin{proof}
Let us consider the tasks that have been computed by all processors
between $t_k(x)$ and $t_k(x+\dx)$. As depicted on
Figure~\ref{fig.evolv}, these tasks can be split into two sets.
\begin{itemize}
\item The first set of tasks consists in those that can be newly
  processed by $P_k$ between $t_k(x)$ and $t_k(x+\dx)$. $P_k$ has the
  possibility to combine the $\dx N$ new elements of $a$ with the $x
  N$ already known elements of $b$ (and to combine the $\dx N$ new
  elements of $b$ with the $x N$ already known elements of $a$). There
  is therefore a total of $2~ x~ \dx~ N^2$ such tasks (at first
  order). Among those, by definition of $g$, the expected number of
  tasks that have not already been processed by other processors is
  given by $2 ~x~ \dx ~g(x)~ N^2$. Therefore, the expected duration of
  this step is given by $t_k(x+\dx) - t_k(x) = \frac{2 ~x ~\dx ~g(x)~
    N^2}{s_k}$.

\item The second set of tasks consists in those computed by other
  processors $P_i,~i \neq k$. Our assumption states
  that we are able to overlap communications by computations (by
  uploading data blocks slightly in advance), so that processors
  $P_i,~i \neq k$ will keep processing tasks between $t_k(x)$ and
  $t_k(x+\dx)$ and will process on expectation $2 ~x~ \dx~ g(x)~
  N^2 \frac{\sum_{i \neq k} s_i }{s_k}$ tasks.

\end{itemize}
Therefore, we are able to estimate how many tasks will be processed
between $t_k(x)$ and $t_k(x+\dx)$ and therefore to compute the
evolution (on expectation) of $g_k$. More specifically, we have

\begin{multline*}
  g_k(x+\dx) ~ \left(1 - (x+\dx)^2 \right) N^2 = \\
  g_k(x) ~ (1 - x^2 ) ~N^2 - 2 ~x ~\dx~ g(x)~ N^2 - 2 ~x~ \dx ~g(x)~ N^2 \frac{\sum_{i \neq k} s_i }{s_k},
\end{multline*}
which gives at first order
$$
g_k(x+\dx) - g_k(x) = g_k(x) ~ \dx ~\frac{- 2 ~x ~\alpha_k}{1-x^2},
$$ where $\alpha_k = \frac{\sum_{i \neq k} s_i}{s_k}$.

Therefore, the evolution of $g_k$ with $x$ is given by the following
ordinary differential equation
$$\frac{g_k'(x)}{g_k(x)} = \frac{- 2 ~x~ \alpha_k}{1-x^2}$$
where both left and right terms are of the form $f'/f$, what leads to
$$\ln (g_k(x)) =  \alpha_k \ln (1-x^2) + K$$ and finally to
$$ g_k(x) =  \exp(K) (1-x^2)^{\alpha_k}, $$
where $exp(K)=1$ since $g_k(0)=1$. This achieves the proof of
Lemma~\ref{lemg}. 
\end{proof}
\medskip

Remember that $t_k(x)$ denotes the time necessary for $P_k$ to know
$x$ elements of $a$ and $b$. Then,

\begin{lemma} 
  \label{lemT} $t_k(x) \sum_i s_i =N^2 (1-  (1-x^2)^{\alpha_k+1})$. 
\end{lemma}

\begin{proof}
  We have seen that some of the tasks that could have been processed
  by $P_k$ (tasks $T_{i,j}$ such that $P_k$ knows both $a_i$ and
  $b_j$) have indeed been processed by other processors. In order to
  prove the lemma, let us denote by $h_k(x)$ the number of such tasks
  at time $t_k(x)$. Then
$$h_k(x+\dx)= h_k(x) + 2~ x~ \dx~ (1 -g_k(x))N^2$$ by definition of $g_k$ so that, using Lemma~\ref{lemg},
$$h_k'(x)=N^2 (2 x - 2 x (1-x^2)^{\alpha_k})$$ and
$$h_k(x)= N^2 (x^2 + \frac{(1-x^2)^{\alpha_k+1}}{\alpha_k+1} + K)$$ and since $h_k(0)=0$,
$$h_k(x)= N^2 (x^2 + \frac{(1-x^2)^{\alpha_k+1}}{\alpha_k+1} - \frac{1}{\alpha_k+1}) .$$

Moreover, at time $t_k(x)$, all the tasks that could have been processed by $P_k$ have
\begin{itemize}
\item either been processed by $P_k$ and there are exactly $t_k(x) s_k$ such tasks since $P_k$ has been  processing all the time in this area,
\item or processed by other processors and there are exactly $h_k(x)$ such tasks by definition of $h_k$.
\end{itemize}
Therefore,
$$x^2 N^2 = h_k(x) + t_k(x) s_k$$
and finally
$$t_k(x) \sum_i s_i =N^2 (1- (1-x^2)^{\alpha_k+1}),$$
which achieves the proof of Lemma~\ref{lemT}.
\end{proof}
\medskip

Above equations well describe the dynamics of \stupidthreshold as long
as it is possible to find blocks of $a$ and $b$ that enable to compute
enough unprocessed tasks. On the other hand, at the end, it is better
to switch to another algorithm, where unprocessed tasks
$T_{i,j}$ are picked up randomly, which possibly requires to send two
blocks $a_i$ and $b_j$. In order to decide when to switch from one
strategy to the other, we introduce an additional parameter $\beta$.

As presented above, a lower bound on the communication volume received
by $P_k$ (if perfect load balancing is achieved) is given by $LB=2 N
\sum_k\sqrt{\mathit{rs}_k}$.  We will switch from the \stupid to
the \simplerandom strategy when the fraction of tasks $x_k^2 N^2$ for
which $P_k$ owns the input data is approximately $\beta$ times what it
would have computed optimally, that is, when $x_k^2$ is close to
$\beta \frac{s_k}{\sum_i s_i} = \beta \mathit{rs}_k$, for a value of
$\beta$ that is to be determined. For the sake of the analysis, it is
important that we globally define the instant at which we switch to
the random strategy, and that it does not depend on the processor
$P_k$. In order to achieve this, we look for $x_k^2$ as $$x_k^2=
(\beta \mathit{rs}_k -\alpha \mathit{rs}_k^2)$$ and we search $\alpha$
such that $t_k(x_k)$ does not depend on $k$ at first order in
$1/\mathit{rs}_k$, where $\mathit{rs}_k$ is of order $1/p$ and $p$
is the number of processors.

\begin{lemma}
\label{lembeta}
If $\alpha=\beta^2/2$, then $$t_k(x_k) \sum_i s_i = N^2 (1 -e^{-\beta} (1+o(\mathit{rs}_k))).$$
\end{lemma}

\begin{proof}
  Since $t_k(x_k) \sum_i s_i =N^2 (1- (1-x_k^2)^{\alpha_k+1},$ then
  \begin{eqnarray*}
    t_k(x_k) &=& \frac{N^2}{\sum_i s_i} ( 1 - e^{\displaystyle \frac{1}{\mathit{rs}_k} \ln (1-\beta \mathit{rs}_k + \alpha \mathit{rs}_k^2))}\\
    & = & \frac{N^2}{\sum_i s_i} ( 1 - e^{\displaystyle
      \frac{1}{\mathit{rs}_k} (-\beta \mathit{rs}_k + \alpha
      \mathit{rs}_k^2- (\beta \mathit{rs}_k)^2/2 ))} \\
    &&\hspace{4cm} \textnormal{ (at first order)}\\
 & = & \frac{N^2}{\sum_i s_i} ( 1 - e^{-\beta}(1+o(\mathit{rs}_k))).
  \end{eqnarray*}
  which achieves the proof of Lemma~\ref{lembeta}.
\end{proof}

One remarkable characteristics of the above result is that it does not
depend (at least up to order 2) on $k$ anymore. Otherwise stated, at
time $T =\frac{N^2}{\sum_i s_i} (1 - e^{\beta}),$ each processor $P_k$
has received $\sqrt{(\beta \mathit{rs}_k -\beta^2/2 \mathit{rs}_k^2)
  N^2} = \sqrt{\beta \mathit{rs}_k} (1 -\beta \mathit{rs}_k/4) N$
data, to be compared with the lower bound on communications for
processor $P_k$: $\sqrt{\mathit{rs}_k} N$.

Using both these results, it is possible to derive the ratio between the overall amount of communication induced by the first phase with respect to the lower bound as a function of $\beta$.

\begin{lemma}
\label{lemP1}
Let us denote by ${\cal V}_{\textsc{Phase1}}$ the volume of the
communications induced by Phase 1 and by $LB = 2 N \sum_k \sqrt{\mathit{rs}_k}$ the lower bound for the communications induced by the whole outer product, then 
$$ \frac{{\cal V}_{\textsc{Phase1}}}{LB} \leq  \sqrt{\beta}  + \frac{\beta^{3/2} \sum_i \mathit{rs}_k^{3/2}}{ 4 LB}   \textnormal{ (at first order)}.
$$
\end{lemma}

\begin{proof} The proof is obtained by replacing ${\cal V}_{\textsc{Phase1}}$ by $\sum_k \sqrt{\beta \mathit{rs}_k} (1 -\beta \mathit{rs}_k/4)  N$.
\end{proof}

Lemma~\ref{lemP1} provides the evaluation of the expected
communication volume induced by the first phase of \stupidthreshold
with respect to the lower bound. In the following, we will establish
a similar result for the second phase in Lemma~\ref{lemP2}.

\begin{lemma}
\label{lemP2}
Let us denote by ${\cal V}_{\textsc{Phase2}}$ the volume of the communications induced by Phase 1 and by $LB = 2 N \sum_k \sqrt{\mathit{rs}_k}$ the lower bound for the communications induced by the whole outer product, then 
$$ \frac{{\cal V}_{\textsc{Phase2}}}{LB} \leq  e^{-\beta} N \frac{1-\sqrt{\beta}  \sum_k \mathit{rs}_k^{3/2}}{\sum_k \mathit{rs}_k^{1/2}}   \textnormal{ (at first order).}
$$
\end{lemma}

\begin{proof}

  During Phase 2, when a processor $P_k$ requests some work, a random
  task is sent among those that have not been processed yet. This task
  $T_{i,j}$ induces either the communication of one data block (if
  either $a_i$ or $b_j$ is already know at $P_k$) or 2 data blocks
  (but not 0 by construction).

  More precisely, since tasks are sent at random and since $P_k$ knows
  a fraction $x_k = \sqrt{\beta \mathit{rs}_k} (1 -\beta
  \mathit{rs}_k/4) $ of the elements of $a$ and $b$ at the end of
  Phase 1,
  \begin{itemize}
  \item a task induces the communication of one block with probability
    $\frac{2 x_k}{1+x_k}$,
  \item a task induces the communication of two blocks with
    probability $\frac{1- x_k}{1+x_k}$.
  \end{itemize}
  so that the expected number of communications per task for $P_k$
  is
  $$ \frac{2 x_k}{1+x_k} \times 1 + \frac{1- x_k}{1+x_k} \times 2
  =\frac{2}{1+x_k}.$$
 Moreover, since Phase 2 starts at the same
  instant on all processors and since processors are continuously
  processing tasks, $P_k$ processes a fraction $\mathit{rs}_k$ of the
  $e^{-\beta} N^2$ remaining tasks. The overall communication cost
  induced by Phase 2 is therefore given (on expectation and at first
  order) by
  $$ {\cal V}_{\textsc{Phase2}} =  e^{-\beta} N^2 \left( 1-\sqrt{\beta}  \sum_k \mathit{rs}_k^{3/2}\right),$$
  which achieves the proof of Lemma~\ref{lemP2}.
\end{proof}

\begin{theorem}
\label{mainthouter}
The ratio of the overall volume of communications to the lower bound if we switch between both phases when $e^{-\beta} N^2 $ tasks remain to be processed is given by 
$$\sqrt{\beta}  + \frac{\beta^{3/2} \sum_k \mathit{rs}_k^{3/2}}{ 4 \sum_k \mathit{rs}_k^{1/2}}  + e^{-\beta} N^2 \frac{1-\sqrt{\beta}  \sum_k \mathit{rs}_k^{3/2}}{\sum_k \mathit{rs}_k^{1/2}}.$$

\end{theorem}

Theorem~\ref{mainthouter} is a direct consequence of Lemma~\ref{lemP1}
and Lemma~\ref{lemP2}. Therefore, in order to minimize the overall
amount of communications, we numerically determine the value of
$\beta$ that minimizes the above expression and then switch between
Phases 1 and 2 when $e^{-\beta} N^2 $ tasks remain to be processed.

\subsection{Assessing the validity of the analysis through simulations}
\label{simuouter}

We have performed simulations to study the accuracy of the previous
theoretical analysis, that is a priori valid only for large values of
$p$ and $N/l$, and to show how it is helpful to compute the threshold
for \stupidthreshold. The simulations have been done using an ad-hoc
event based simulation tool, where processors request new tasks as
soon as they are available, and tasks are allocated based on the given
runtime dynamic strategy. Again, processor speeds are chosen uniformly
in the interval $[10, 100]$.  This degree of heterogeneity may seem
excessive but we show in Section~\ref{hetero-models} that using a
different heterogeneity model does not significantly impact the
results.  The communication amount of each strategy is normalized by
the lower bound computed in
Section~\ref{outer_analysis}. Figure~\ref{fig.simubasicP100} presents
the results for vectors of 100 blocks and
Figure~\ref{fig.simubasicP1000} does the same for vectors of 1000
blocks.

\begin{figure}[htbp]
  \centering
  \includegraphics[width=0.7\linewidth]{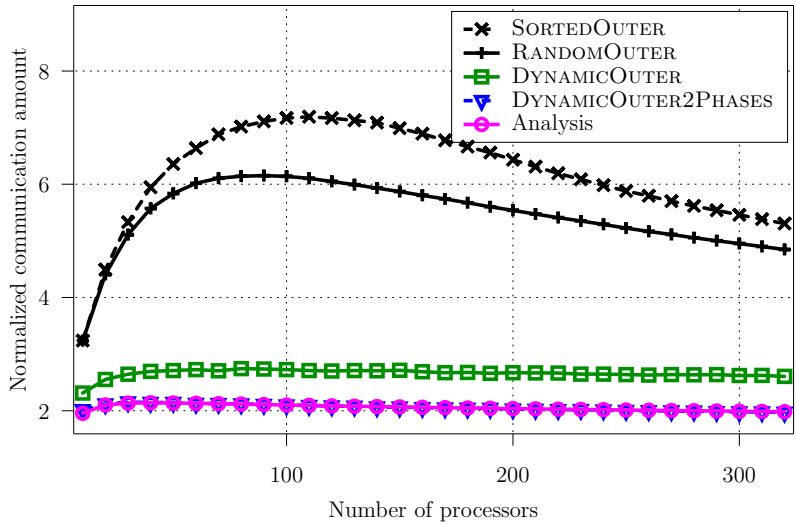}
  \caption{Communication amounts of all outer-product strategies for vectors of size $N/l=100$ blocks ($(N/l)^2$ tasks).}
  \label{fig.simubasicP100}
\end{figure}

\begin{figure}[htbp]
  \centering
  \includegraphics[width=0.7\linewidth]{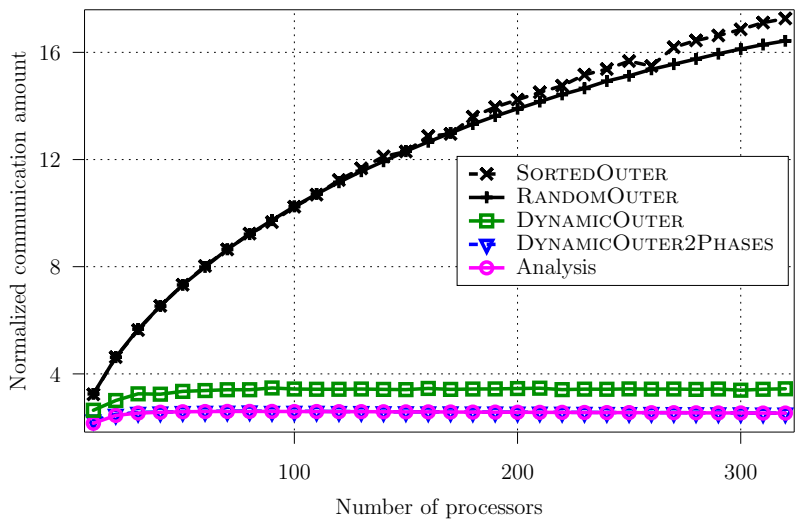}
  \caption{Communication amounts of all outer-product strategies for vectors of size $N/l=1000$ blocks ($(N/l)^2$ tasks).}
  \label{fig.simubasicP1000}
\end{figure}

In both figures, the analysis is extremely close to the performance of
\stupidthreshold (which makes them indistinguishable on the figures)
and proves that our analysis succeed to accurately model our dynamic
strategy, even for relatively small values of $p$ and $N/l$. Moreover, we can see in Figure~\ref{fig.simubasicP1000} that it
is even more crucial to use a data-aware dynamic scheduler when $N$ is
large, as the ratio between the communication amount of simple random
strategies (\simplerandom and \simplesorted) and dynamic data-aware
schedulers (such as \stupidthreshold) can be very large.

Our second objective is to show that the theoretical analysis that we propose can be used in order to accurately
compute the threshold of \stupidthreshold, \ie, that the $\beta$
parameter computed earlier is close to the best one. To do this, we
compare the communication amount of \stupidthreshold for various
values of the $\beta$ parameter. Figure~\ref{fig.beta} shows the
results for 20 processors and $N/l=100$. This is done for a single and
arbitrary distribution of computing speeds, as it would make no sense
to compute average values for different distributions since they would lead to
different optimal values of $\beta$. This explains the irregular
performance graph for \stupidthreshold.  This figure shows that in the
domain of interest, \ie for $3\leq \beta \leq 6$, the analysis
correctly fits to the simulations, and that the value of $\beta$ that
minimizes the analysis (here $\beta = 4.17$) lies in the interval of
$\beta$ values that minimize the communication amount of
\stupidthreshold. To compare to Figure~\ref{fig.threshold}, this
corresponds to 98.5\% of the tasks being processed in the first phase.

\begin{figure}[htbp]
  \centering
  \includegraphics[width=0.7\linewidth]{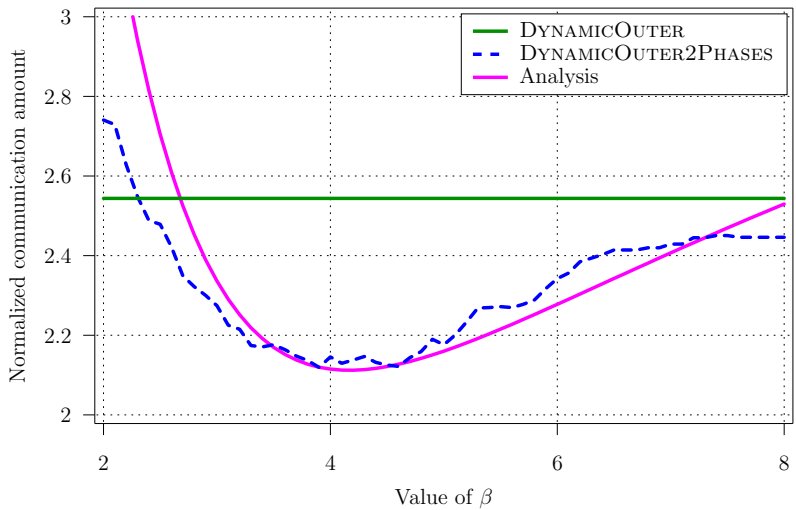}
  \caption{Communication amounts of \stupidthreshold and its analysis
    for varying value of the $\beta$ parameter which defines the
    threshold.}
  \label{fig.beta}
\end{figure}

\subsection{Impact of the heterogeneity}
\label{hetero-models}

The speed distribution used in the previous experiments (speeds taken
in the interval $[10,100]$) may seem too heterogeneous to  reasonably
model actual computing platforms, where heterogeneity comes either
from the use of a few classes of different processors (new and old
machines, processor equipped with accelerators or not, etc.) or from
the fact that machines are not dedicated, which implies stochastically
variable processor speed. It is natural to ask whether the speed
distribution impacts the ranking of the previous heuristics, or the
accuracy of our analysis.

\begin{figure}[htbp]
  \centering
  \includegraphics[width=0.7\linewidth]{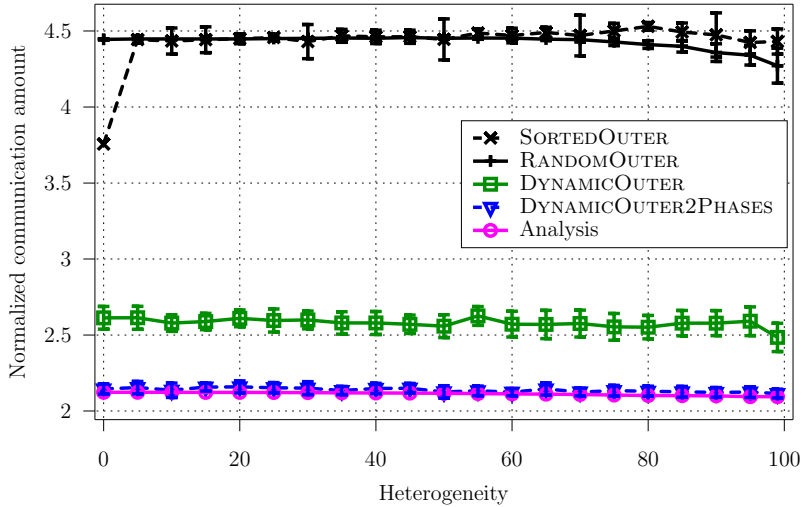}
  \caption{Behavior of the heuristics for outer product for different
    values of heterogeneity ($p=20$ processors and $N/l=100$ blocks). For a given value
    $h$ of heterogeneity, processor speeds are taken uniformly at
    random in the interval $[100-h, 100+h]$. }
  \label{fig.hetero}
\end{figure}

Figure~\ref{fig.hetero} presents the behavior of all previous
heuristics for a varying range of heterogeneity. A heterogeneity of 0
means perfectly homogeneous computing speeds, while a heterogeneity of
100 means that the ratio between the smallest and the largest speeds
is large. In this figure and the following one, error bars represents
the standard deviations with 50 tries. We notice that the
heterogeneity degree has very little impact on the relative amounts of
communication of the studied heuristics.

\begin{figure}[htbp]
  \centering
  \includegraphics[width=0.7\linewidth]{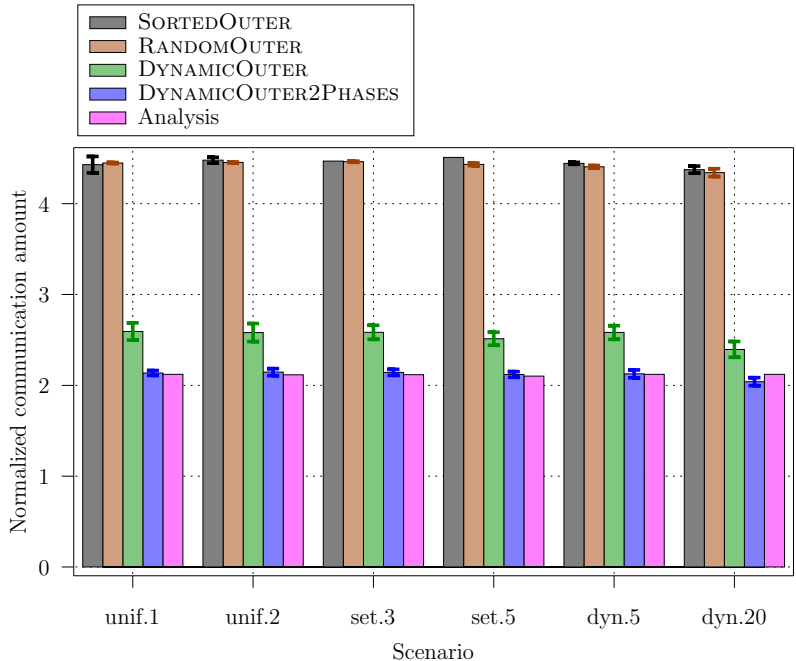}
  \caption{Behavior of the heuristics for outer product for different scenarios
    of heterogeneity ($p=20$ processors and $N/l=100$ blocks).}
  \label{fig.hetero2}
\end{figure}

In Figure~\ref{fig.hetero2}, we study the same heuristics using
different scenarios:
\begin{itemize}
\item Scenarios unif.1 and unif.2 corresponds to the previous setting,
  with speeds taken uniformly at random in intervals $[80,120]$
  (unif.1) and $[50,150]$ (unif.2).
\item Scenarios set.3 and set.5 corresponds to the case when there are
  a few classes of processors with different speed. The speeds are
  then taken uniformly from the set of possible speeds: (80, 100, 150)
  for set.3 or (40, 80, 100, 150, 200) for set.5.
\item Scenarios dyn.5 and dyn.20 corresponds to very simple dynamic
  settings. Each computing speed is first taken uniformly at random in
  interval $[80,120]$. Then, after computing a task, a processor sees
  its computing speed randomly changed by up to 5\% (dyn.5) or 20\%
  (dyn.20). 
\end{itemize}
This figure shows that neither the speed distribution nor the dynamic
evolution of the speeds notably affect the performance of the
heuristics.

\subsection{Runtime estimation of $\beta$}
\label{betaestimation}

In order to estimate the $\beta$ parameter in the \stupidthreshold
strategy, it seems necessary to know the processing speed, as $\beta$
depends on $\sum_k \sqrt{{s_k}/{\sum_i s_i}}$. However, we have
noticed a very small deviation of $\beta$ with the speeds. For
example, in Figure~\ref{fig.beta}, the value of $\beta$ computed when
assuming homogeneous speeds (4.1705) is very close to the one computed
for heterogeneous speeds (4.1679). 

For a large range of $N/l$ and $p$ values (namely, $p$ in $[10,1000]$
and $N/l \in [\max(10,\sqrt{p}), 1000]$), for processor speeds in
$[10,100]$, the optimal value for $\beta$ goes from 1 to 6.2. However,
for fixed values of $N/l$ and $p$, the deviations among the $\beta$
values obtained for different speed distributions is at most 0.045
(with 100 tries).  Our idea is to approximate $\beta$ with
$\beta_{\mathit{hom}}$ computed using a homogeneous platform with the
same number of processors and with the same matrix size. The relative
difference between $\beta_{\mathit{hom}}$ and the average $\beta$ of
the previous set is always smaller than 5\%. Moreover, the error on
the communication volume predicted by the analysis when using
homogeneous speeds instead of the actual ones is at most 0.1\%. These
figures are derived with the most heterogeneous speed distribution
(speeds in $[10,100]$) and thus hold for the other distributions of
Section~\ref{hetero-models} as well.

This proves that even if our previous analysis ends up with a formula
for $\beta$ that depends on the computing speeds, in practice, only
the knowledge of the matrix size and of the number of processors are
actually needed to define the threshold $\beta$. Our dynamic scheduler
\stupidthreshold is thus totally agnostic to processor speeds.

\section{Matrix Multiplication}
\label{matrix_mult}

We adapt here the previous dynamic algorithm and its theoretical
analysis to a more complex problem: the multiplication of two
matrices.


\subsection{Notations and dynamic strategies}

We first adapt the notations to the problem of computing the product
of two matrices $C = A B$. As in the previous section, we consider
that all transfers and computations are performed using blocks of size
$l\times l$, so that all three matrices are composed of $N^2/l^2$
blocks and $A_{i,j}$ denotes the block of $A$ on the $i$th row and the
$j$th column. The basic computation step is a task $T_{i,j,k}$, which
corresponds to the update $C_{i,j} \gets C_{i,j} + A_{i,k}
B_{k,j}$. To perform such a task, a processor has to receive the input
data from $A$ and $B$ (of size $2l^2$), and to send the result (of
size $(N/l)^2$) back to the master at the end of the
computation. Thus, it results in a total amount of communication of
$3~(N/l)^2$. As previously, in order to minimize the amount of
communications, our goal is to take advantage of the blocks of $A$,
$B$ and $C$ that have already been sent to a processor $P_u$ when
allocating a new task to $P_u$. Note that at the end of the
computation, all $C_{i,j}$s are sent back to the master that computes
in turn the final results by adding the different contributions. This
computational load is much smaller than computing the products
$T_{i,j,k}$ and we will neglect it in what follows.

As we assume that processors work during the whole process, the load
imbalance, \ie the difference between the amount of work processed by
$P_i$ and what it should have processed given its speed is at most one
block. Thus, a maximal block size $l$ can easily be derived from a
maximal load imbalance. The value of $l$ must also be large enough to
overlap communications of size $3 l^2$ with computations of size
$l^3$.  As usual, the block size should also be large enough to
benefit from BLAS effect and small enough so as to fit into caches. We
assume that the optimal block size $l$ is computed by the runtime
environment.
\medskip

The simple strategies \simplerandom and \simplesorted translate very
easily for matrix multiplication into the strategies \simplerandommat
and \simplesortedmat. We adapt the \stupid strategy into \stupidmat as
follows. We ensure that at each step, for each processor $P_u$ there
exist sets of indices $I$, $J$ and $K$ such that $P_u$ owns all values
$A_{i,k}$, $B_{k,j}$, $C_{i,j}$ for $i\in I$, $j\in J$ and $k\in K$,
so that it is able to compute all corresponding tasks
$T_{i,j,k}$. When a processor becomes idle, instead of sending a
single block of $A$, $B$ and $C$, we choose a tuple $(i,j,k)$ of new
indices (with $i\notin I$ , $j\notin J$ and $k\notin K$) and send to
$P_u$ all the data needed to extend the sets $I,J,K$ with
$(i,j,k)$. This corresponds to sending $3\times(2\card{I}+1)$ data blocks to
$P_u$ (note that $\card{I} = \card{J} = \card{K}$). In fact, blocks of
$C$ are not send by the master to the processor, but on the contrary
will be sent back to the master at the end of the computation; however,
this does not change the analysis since we are only interested in the
overall volume of communications. Then, processor $P_u$ is allocated
all the unprocessed tasks that can be done with the new
data. Algorithm~\ref{algo.stupidmat} details this strategy.

\begin{algorithm}[htbp]
  \DontPrintSemicolon
  \While{there are unprocessed tasks}{    
    Wait for a processor $P_u$ to finish its task\;
    $I \gets \{i\textnormal{~such that~}P_u\textnormal{~owns~}A_{i,k} \textnormal{~for some~} k\} $\;
    $J \gets \{i\textnormal{~such that~}P_u\textnormal{~owns~}B_{k,j} \textnormal{~for some~} k\} $\;
    $K \gets \{i\textnormal{~such that~}P_u\textnormal{~owns~}A_{i,k} \textnormal{~for some~} i\} $\;

    Choose $i\notin I$ , $j\notin J$ and $k\notin K$ uniformly at random\;
    Send the following data blocks to $P_u$:
    \begin{itemize}
    \item  $A_{i,k'}$ for $k'\in K \cup \{ k \}$ and $A_{i',k}$ for\\ $i' \in I \cup \{i\}$ 
    \item  $B_{k,j'}$ for $j'\in J \cup \{j\}$ and $B_{k',j}$ for\\ $k' \in K \cup \{k\}$
    \item  $C_{i,j'}$ for $j'\in J \cup \{j\}$ and $C_{i',j}$ for\\ $i' \in I \cup \{i\}$ 
    \end{itemize}

    Allocate all tasks $\{T_{i',j',k'}$with $i' = i $ or $j' = j$ or $k' = k \}$ that are not yet processed to $P_u$ and mark them processed\;
  }  
  \caption{\stupidmat strategy.}
  \label{algo.stupidmat}
\end{algorithm}

As in the case of the outer product, when the number of remaining
blocks to be processed becomes small,\simplerandommat strategy
outperforms the \stupidmat strategy. Therefore, we introduce the
intermediate \stupidthresholdmat strategy that consists into two
phases. During Phase 1, the \stupidmat strategy is used. Then, when
the number of remaining tasks becomes smaller than $e^{-\beta} N^3$
for a value of $\beta$ that is to be determined, we switch to Phase 2
and use strategy \simplerandommat. As in the case of the outer
product, the theoretical analysis proposed in the next section will
help us to determine the optimal value of $\beta$, \ie the instant
when to switch between phases in order to minimize the overall
communication volume in the \stupidthresholdmat strategy.

\subsection{Theoretical analysis of dynamic randomized strategies}

In this section, our aim is to provide an analytical model for
Algorithm \stupidthresholdmat similarly to what has been done for Algorithm \stupid in
Section~\ref{outer_analysis}. The analysis of both processes is in
fact rather similar, so that we will mostly state the corresponding
lemmas, the proofs being similar to those presented in
Section~\ref{outer_analysis}.

In what follows, we will assume that $N$, the size of matrices $A$, $B$
and $C$, is large and we will consider a continuous dynamic process
whose behavior is expected to be close to the one of \stupidthresholdmat. In
what follows, as in Section~\ref{outer_analysis}, we will concentrate
on processor $P_k$ whose speed is $s_k$ and relative speed $\mathit{rs}_k=\frac{s_k}{\sum_i s_i}$. We will also denote by $C=A
\times B$ the result of the matrix multiplication. Note that throughout this section,
$A_{i,k}$ denotes the \emph{element} of $A$ on the $i$th row and $j$th
column.

Let us assume that there exist 3 index sets $I,J$ and $K$ such that 
\begin{itemize}
\item $P_k$ knows all elements $A_{i,k}$, $B_{k,j}$ and $C_{i,j}$ for
  any $(i,j,k)\in I\times J \times K$.
\item $I, J$ and $K$ have size $y$.
\end{itemize}

In Algorithm \stupidthresholdmat, at each step, $P_k$ chooses to increase its
knowledge by increasing $y$ by $l$, which requires to receive $(2  y+1)l$
elements of each $A$, $B$ and $C$. As we did in
Section~\ref{outer_analysis}, we will concentrate on $x=y/N$, and
assuming that $N$ is large, we will change the discrete process into a
continuous process described by an ordinary differential equation
depicting the evolution of expected values and we will rely on
extensive simulations to assert that this approximation is valid.

In this context, let us consider that an elementary task $T(i,j,k)$
consists in computing $C_{i,j} \gets C_{i,j} +A_{i,k} B_{k,j}$. There are $N^3$ such
tasks. In what follows, we will denote by $g_k(x)$ the fraction of
elementary tasks that have not been computed yet at the instant when
$P_k$ knows $x^2$ elements of $A, B$ and $C$ respectively, in the
computational domain that does not include the tasks $T(i,j,k)$ such
that $(i,j,k)\in I\times J \times K$ (this domain is equivalent to the
``L''-shaped area for the outer product in Section~\ref{outer_analysis}). The following lemma enables to
understand the dynamics of $g_k$ (all proofs are omitted because they
are very similar to those of Section~\ref{outer_analysis}).

\begin{lemma}
\label{lemgm}
$g_k(x)=(1-x^3)^{\alpha_k}$, where $\alpha_k = \frac{\sum_{i \neq k} s_i}{s_k}$.
\end{lemma}

Let us now denote by $t_k(x)$ the time step such that index sets $I$, $J$ and $K$ have size $x$.
Then,
\begin{lemma}
\label{lemTm}
$t_k(x) \sum_i s_i =1- N^2 (1- (1-x^3)^{\alpha_k+1})$.
\end{lemma}

Above equations well describe the dynamics of \stupidthresholdmat as long as it
is possible to find elements of $A$, $B$ and $C$ that enable to
compute enough unprocessed elementary tasks. On the other hand, as in
the case of \stupidthreshold, at the end, it is better to switch to another
algorithm, where unprocessed elementary tasks $T(i,j,k)$ are
picked up randomly, what requires possibly to send all three values of
$A_{i,k}$, $B_{k,j}$ and $C_{i,j}$. In order to decide when to switch
from one strategy to the other, let us introduce the additional
parameter $\beta$.

As in the outer-product problem, a lower bound on the communication
volume received by $P_k$ can be obtained by considering that each
processor has a cube of tasks $T_{i,j,k}$ to compute, proportional to
its relative speed. The edge-size of this cube is thus $N \sqrt[3]{\mathit{rs}_k}$.
To compute all tasks in this cube, $P_k$ needs to receive a square of
each matrix, that is $3 N^2 \mathit{rs}_k^{2/3}$. 

In order to determine when we should switch between Phase 1 and Phase 2, we
can observe that if $x_k^3= \beta \mathit{rs}_k - \beta^2/2
\mathit{rs}_k^2$, then 
$$
t_k(x_k) \sum_i s_i = N^2 ( 1 -e^{-\beta}(1+o(\mathit{rs}_k))),
$$
so that at first order, $t_k(x_k)$ is independent of $k$. The instant
$t=\frac{N^2}{\sum_i s_i} ( 1 - e^{-\beta})$ is therefore chosen to
switch between Phases 1 and 2.

As in the context of the outer product, we need to find the value of $\beta$ that minimizes the volume of communications. If the switch occurs at time $t=\frac{N^2}{\sum_i s_i} ( 1 - e^{-\beta})$, then
\begin{itemize}
\item the volume of communications during Phase 1 is given by $$3 N^2 \beta^{2/3} \sum_k \mathit{rs}_k^{2/3} - 3 N^2 \beta^{5/3} \sum_k \mathit{rs}_k^{5/3},$$
\item the volume of communications during Phase 2 is given by $$ e^{-\beta} N^3 \left(1 - \beta^{2/3} \sum_k \mathit{rs}_k^{5/3}\right),$$
\end{itemize}
so that the total amount of communications with respect to the lower bound $3 N^2 \sum \mathit{rs}_k^{2/3}$ is given by
$$ \beta^{2/3} - \beta^{5/3} \frac{\sum_k \mathit{rs}_k^{5/3}}{\sum_k \mathit{rs}_k^{2/3}} +  \frac{e^{-\beta} N}{\sum_k \mathit{rs}_k^{5/3}} \left(1 - \beta^{2/3} \sum_k \mathit{rs}_k^{5/3}\right).$$


\subsection{Simulation Results}


We have conducted extensive simulations to compare the performance of
the dynamic strategies with the previous analysis.
Figure~\ref{fig.simumat40} presents the results for matrices of size
40x40 and Figure~\ref{fig.simumat100} presents the results for matrices of size
100x100. As in previous simulations, processor speeds are chosen
uniformly at random in the interval $[10, 100]$ and all amounts of
communications have been normalized using the lower bound $3 N^2 \sum_k \mathit{rs}_k^{2/3}$ on
communications presented in the previous section.


\begin{figure*}[htbp]
  \centering
  \includegraphics[width=0.7\linewidth]{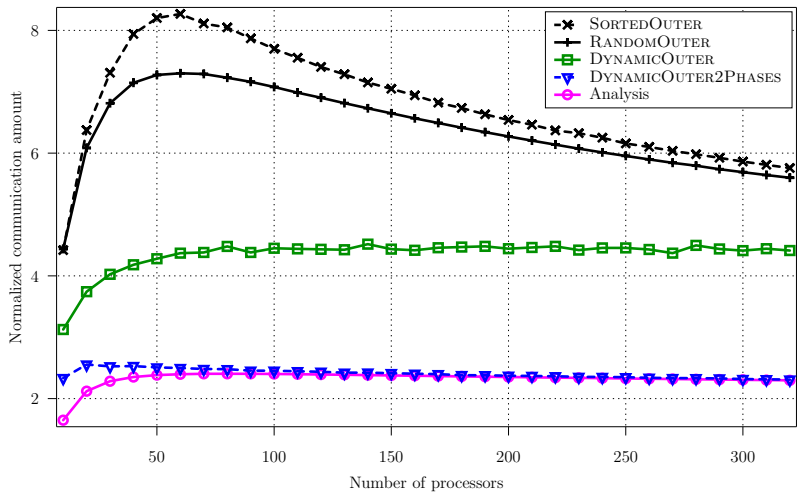}
  \caption{Communication amounts of all strategies for matrices of size $N/l=40$
    blocks ($N^3/l^3 = 64,000$ tasks).}
  \label{fig.simumat40}
\end{figure*}

\begin{figure*}[htbp]
  \centering
  \includegraphics[width=0.7\linewidth]{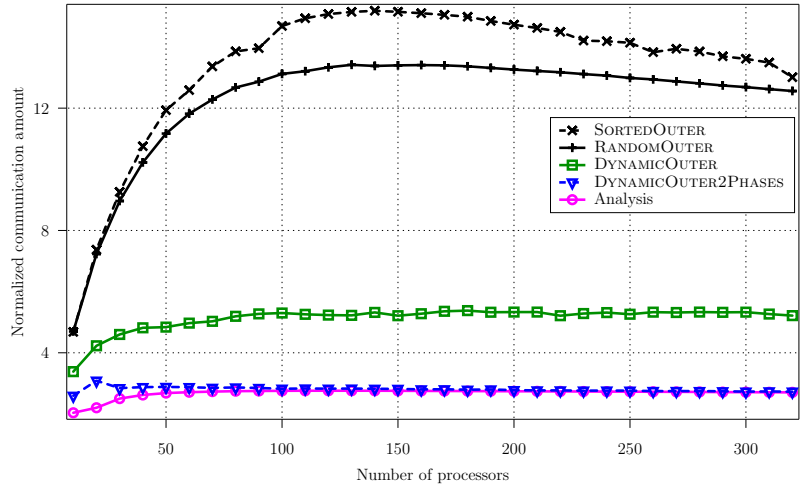}
  \caption{Communication amounts of all strategies for matrices of size $N/l=100$
    blocks ($N^3/l^3 = 1,000,000$ tasks).}
  \label{fig.simumat100}
\end{figure*}

As for the outer-product problem, we notice that data-aware strategies
largely outperform simple strategies, and that \stupidthresholdmat is
able to reduce the communication amount even more than \stupidmat. When
the number of processors is large enough (\ie in our simulation
setting, $p\geq 50$), our previous analysis is able to very accurately
predict the performance of \stupidthresholdmat.

\begin{figure}[htbp]
  \centering
  \includegraphics[width=0.7\linewidth]{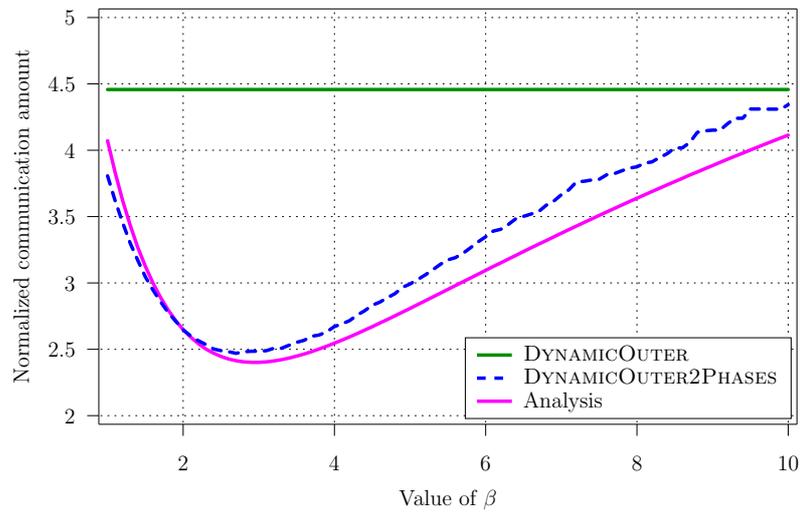}
  \caption{Communication amount of \stupidthresholdmat and its
    analysis for varying value of the $\beta$ parameter which defines
    the threshold.}
  \label{fig.betamat}
\end{figure}

We also performed simulations of \stupidthresholdmat with varying values
of $\beta$ to check if the optimal value determined in the theoretical analysis actually
minimizes the amount of communications. This is illustrated in
Figure~\ref{fig.betamat}, for 100 processors, $N/l=40$ and a fixed
distribution of computing speeds. As for the outer product, we notice
that the analysis accurately models the amount of communications of
\stupidthresholdmat in the range of values of interest of $\beta$, and
that the optimal value of $\beta$ for the analysis (2.95) allows to
obtain an amount of communications that is close to optimal. This
corresponds to 94.7\% of the tasks to be processed by the first phase
of the algorithm.  As for the outer product, we also notice that the
value of $\beta$ given by an analysis which is agnostic to processor
speeds and assumes homogeneous speeds is very close to the optimal
value (2.92 on this example).

\section{Conclusion and perspectives}
\label{conclusion}

The contributions of this paper follow two directions. First, we have
proposed randomized dynamic scheduling strategies for the outer
product and the matrix multiplication kernels. We have proved that
dynamic scheduling strategies that aim to place tasks on processors
such that the induced amount of communications is as small as possible
perform well. Second, we have been able to propose an Ordinary
Differential Equation (ODE) whose solution describes very well the dynamics
of the system. Even more important, we prove that the analysis of the
dynamics of the ODE can be used in order to tune parameters and to inject
some static knowledge which is useful to increase the efficiency of
dynamic strategies.

A lot remains to be done in this domain, that we consider as crucial
given the practical and growing importance of dynamic runtime
schedulers. First, it would be of interest to be able to provide
analytical models for a larger class of dynamic schedulers even in the
case of independent tasks, and to analyze their behavior also in
dynamic environments (when the performance of the resources is unknown
and varies over time). Then, it would be very useful to extend the
analysis to applications involving both data and precedence
dependencies. Extending this work to regular dense linear algebra
kernels such as Cholesky or QR factorizations would be a promising first step in
this direction.


\section{Acknowledgement}

This work has been partially supported by the {\em ANR SOLHAR}
project, funded by the French Research Agency.


\end{document}